\documentclass[
    aps,
    pra,
    reprint,
    amsmath,
    amssymb,
    superscriptaddress,
    nofootinbib
]{revtex4-2}

\usepackage{graphicx}
\usepackage{dcolumn}
\usepackage{bm}
\usepackage[utf8]{inputenc}
\usepackage[T1]{fontenc}
\usepackage{mathptmx}
\usepackage{hyperref}
\usepackage{amsmath}
\usepackage{amssymb}
\usepackage{amsthm}

\hypersetup{
    setpagesize = false,
    colorlinks  = true,
    urlcolor    = blue,
    linkcolor   = black,
    citecolor   = black
}

\newtheorem{theorem}{Theorem}
\newtheorem{lemma}{Lemma}
\newtheorem{proposition}{Proposition}
\newtheorem{corollary}{Corollary}
\newtheorem{remark}{Remark}

\newcommand{\id}{\operatorname{id}}
\newcommand{\EB}{\mathrm{EB}}
\newcommand{\EA}{\mathrm{EA}}
\newcommand{\Sep}{\mathrm{Sep}_{+}}
\newcommand{\Dcal}{\mathcal{D}}
\newcommand{\Mtwo}{\mathcal{B}(\mathbb{C}^{2})}

\begin{document}

\title{Feasibility Ordering of Entanglement-Source Placement for Qubit Channels}

\author{Samuel A. M\'arquez Gonz\'alez}

\affiliation{
Department of Mathematical Sciences,
Rutgers University,
Camden, New Jersey 08102, USA
}

\email{sam959@scarletmail.rutgers.edu}

\begin{abstract}
This work studies the placement of an entanglement source along a communication line formed by two noisy qubit channels. Recent work argued, on analytical and numerical grounds, that midpoint placement should be at least as favorable as endpoint placement. Here it is shown that, for arbitrary qubit channels, if a sequential composition can preserve entanglement, then the corresponding parallel action cannot annihilate all entanglement. The proof uses the transpose-factorization criterion introduced in that recent work. Quantum Sinkhorn scaling converts every strictly positive qubit channel into a unital representative, and the special normal form of unital qubit channels then yields the required factorization of the transposed map through the original channel. Depolarizing regularization and the closedness of the set of entanglement-breaking channels extend the result to arbitrary channels. Consequently, $\Lambda_1 \otimes \Lambda_2 \in \mathrm{EA}$ implies $\Lambda_2 \circ \Lambda_1 \in \mathrm{EB}$, and, by exchanging the two channels, the same holds for the opposite composition order. This proves the recent conjecture that midpoint placement is optimal for all qubit channels, in the feasibility sense in which that optimality was originally defined.
\end{abstract}

\maketitle

\section{Introduction}

Entanglement is a central resource in quantum information science, underpinning communication, cryptography, teleportation, distributed computation, and networked quantum protocols \cite{HorodeckiReview2009,Streltsov2012QuantumCost,Chuan2012Discord,Cubitt2003Separable,Zuppardo2016Excessive}. A basic architectural problem therefore involves both how to generate entanglement and where to generate it when noisy quantum channels connect two parties. This question belongs to a broader line of work on one-shot entanglement transmission, noisy entanglement sharing, and the role of pre-existing correlations in distribution protocols \cite{Pal2014Singlet,Streltsov2015Unified,SiddhuSmolin2023Optimal}.

The structure of qubit channels makes this setting especially attractive analytically. Qubit channels admit useful affine and canonical parametrizations \cite{FujiwaraAlgoet1999,RuskaiSzarekWerner2002,BraunEtAl2014}, while the concatenation of channels has a well-developed algebraic and dynamical theory \cite{WolfCirac2008}. At the same time, the distinction between entanglement breaking (EB) and entanglement annihilating (EA) is essential. EB channels destroy entanglement between the transmitted system and every reference system, whereas EA bipartite channels destroy entanglement internal to the bipartite system on which they act \cite{MoravcikovaZiman2010,HorodeckiShorRuskai2003,Ruskai2003QubitEB}. Local two-qubit EA channels and general criteria for EA maps have been developed in Refs.~\cite{FilippovRybarZiman2012,FilippovZiman2013,LamiHuber2016}, with more recent extensions based on Schmidt number and ordered-vector-space methods \cite{AubrunMullerHermes2023,MallickGangulyMajumdar2026,LaPianaMullerHermes2026}.

Masajada, Fellous-Asiani, and Streltsov recently compared two source placements for a line formed by qubit channels $\Lambda_1$ and $\Lambda_2$ \cite{MasajadaFellousStreltsov2026}. With the source at the midpoint, the two halves of an entangled input traverse separate channels,
\begin{equation}
  \Lambda_{\mathrm{mid}} = \Lambda_1 \otimes \Lambda_2 .
  \label{eq:mid}
\end{equation}
With the source at an endpoint, one half is retained locally while the other passes through the channels in series,
\begin{equation}
  \Lambda_{\mathrm{edge}} = \mathrm{id} \otimes (\Lambda_2 \circ \Lambda_1) .
  \label{eq:edge}
\end{equation}
These authors conjectured that midpoint placement is optimal for all qubit channels. Optimality is meant there in a feasibility sense, in which one placement is better than another if entanglement surviving the second for some input guarantees that entanglement survives the first for some input. The same work also proposes a stronger quantitative conjecture, supported by numerical tests on randomly generated channels, that compares the minimum eigenvalue of the partially transposed Choi matrix of the sequential channel with an SDP bound for the parallel channel. That inequality is not addressed here. The present work proves the feasibility conjecture, namely that no pair of qubit channels allows an endpoint placement to preserve some entanglement while the midpoint placement annihilates all of it. In channel language,
\begin{equation}
  \Lambda_2 \circ \Lambda_1 \notin \mathrm{EB}
  \;\Longrightarrow\;
  \Lambda_1 \otimes \Lambda_2 \notin \mathrm{EA} .
  \label{eq:qualitative}
\end{equation}

Ref.~\cite{MasajadaFellousStreltsov2026} already provides two important routes toward Eq.~\eqref{eq:qualitative}. Its unital-channel result permits an arbitrary completely positive intermediate filter, and a second criterion proves Eq.~\eqref{eq:qualitative} whenever the channel transpose admits a factorization
\begin{equation}
    \Lambda_1^{T}
    =\mathcal{F}\circ\Lambda_1\circ\mathcal{E},
    \label{eq:source-factorization-preview}
\end{equation}
with \(\mathcal{F}\) positive and \(\mathcal{E}\) completely positive. The same work constructs this factorization almost everywhere for qubit channels of Kraus rank at most three. The argument below extends that route: quantum Sinkhorn scaling produces the factorization \eqref{eq:source-factorization-preview} for every strictly positive qubit channel, without a Kraus-rank restriction, and a depolarizing regularization then reaches the boundary of the full channel set.

This strategy connects naturally with quantum Sinkhorn scaling and filter normal forms \cite{Sinkhorn1964,Gurvits2004,GeorgiouPavon2015,Cariello2019,Filippov2021Sinkhorn}, invertible local-filter methods for entanglement robustness \cite{VerstraeteDehaeneDeMoor2001,VerstraeteDehaeneDeMoor2003,FilippovFrizenKolobova2018}, and the broader literature on how channel composition enters EB classes \cite{LamiGiovannetti2015Indices,LamiGiovannetti2016Saving,KennedyManorPaulsen2018,ChristandlMullerHermesWolf2019,ChenYangTang2019,SinghNechita2022,RahamanJaquesPaulsen2018,HansonRouzeFranca2020,AhiableEtAl2021}. The relation established here is different from an iteration theorem: ordinary physical EA of a parallel pair is converted directly into EB of the ordered sequential composition.

Section~II fixes the channel, separability, transposition, filtering, and scaling conventions. Section~III establishes the transpose factorization for strictly positive qubit channels and extends the resulting implication to arbitrary channels. Section~IV translates the theorem into the source-placement setting. Section~V compares the result with nearby EA/EB, Lorentz-cone, and channel-composition results. Section~VI isolates the precise higher-dimensional obstruction, followed by the conclusion in Sec.~VII.

\section{Preliminaries}

\subsection{Quantum channels, Choi states, and separability}

Consider completely positive trace-preserving (CPTP) maps
\begin{equation}
    \Lambda:\Mtwo\rightarrow\Mtwo.
\end{equation}
The Choi--Jamio\l kowski representation identifies complete positivity with positivity of a bipartite operator associated with the map \cite{Jamiolkowski1972,Choi1975}. For a qubit channel, define the normalized Choi state
\begin{equation}
    J(\Lambda)
    = (\id\otimes\Lambda)
      \left(\lvert\Phi^+\rangle\!\langle\Phi^+\rvert\right),
\end{equation}
where $\mathrm{id}$ denotes the identity channel and
\begin{equation}
    \lvert\Phi^+\rangle
    =\frac{1}{\sqrt{2}}(\lvert00\rangle+\lvert11\rangle).
\end{equation}
A channel is EB if and only if its Choi state is separable \cite{HorodeckiShorRuskai2003,Ruskai2003QubitEB}. For two-qubit states, separability is equivalent to positivity under partial transposition (PPT) \cite{Peres1996,Horodecki1996}. Negativity provides a computable entanglement monotone based on the spectrum of the partial transpose \cite{VidalWerner2002}; it is useful in numerical studies of source placement, although it is not needed in the proof below.

It is convenient to work with the cone of unnormalized separable positive operators,
\begin{equation}
\Sep(A{:}B)
=\left\{
\sum_i X_i\otimes Y_i:\; X_i\ge0,\;Y_i\ge0
\right\}.
\label{eq:sepcone}
\end{equation}
A normalized bipartite state is separable exactly when it belongs to this cone and has unit trace. This homogeneous formulation avoids unnecessary normalization factors when completely positive maps are not trace preserving.

A bipartite CP map \(\Psi\) is called EA if
\begin{equation}
    \Psi(X)\in\Sep(A{:}B)
    \quad\text{for every }X\ge0.
\label{eq:EAdef}
\end{equation}
For a CPTP map, Eq.~\eqref{eq:EAdef} is equivalent to the usual definition involving all density operators \cite{MoravcikovaZiman2010,FilippovRybarZiman2012,FilippovZiman2013}. A single-system channel \(\Lambda\) is EB if \(\id\otimes\Lambda\) is EA. In finite dimensions, the Choi characterization makes the EB set a closed convex set \cite{HorodeckiShorRuskai2003,AhiableEtAl2021}.

\subsection{Local positive maps and completely positive filters}

For an operator \(A\), define the completely positive map
\begin{equation}
    \Phi_A(X)=AXA^{\dagger}.
\label{eq:filter}
\end{equation}
If \(A\) is invertible, then \(\Phi_A\) is an invertible linear map and its inverse is the CP map \(\Phi_{A^{-1}}\). These maps need not preserve the trace. Such local filtering operations are standard in the classification and manipulation of bipartite entanglement \cite{VerstraeteDehaeneDeMoor2001,VerstraeteDehaeneDeMoor2003} and in filter-normal-form arguments for noisy qubit dynamics \cite{FilippovFrizenKolobova2018}.

\begin{lemma}[Separable-cone stability]
\label{lem:sepstability}
If \(\mathcal{F}_A\) and \(\mathcal{F}_B\) are positive linear maps, then
\begin{equation}
(\mathcal{F}_A\otimes\mathcal{F}_B)\big[\Sep(A{:}B)\big]
\subseteq \Sep(A'{:}B').
\end{equation}
Consequently, local CP prefilters and local positive postmaps preserve the separability conclusion whenever the intermediate operator is separable.
\end{lemma}

\begin{proof}
Every element of \(\Sep(A{:}B)\) is a finite sum \(\sum_i X_i\otimes Y_i\) with \(X_i,Y_i\ge0\). Positivity of the local maps gives \(\mathcal{F}_A(X_i)\ge0\) and \(\mathcal{F}_B(Y_i)\ge0\) term by term, so the image remains a finite sum of positive product operators.
\end{proof}

\subsection{Strict positivity and quantum Sinkhorn scaling}

A linear map \(\Lambda\) is strictly positive if
\begin{equation}
X\ge0,\quad X\neq0
\quad\Longrightarrow\quad
\Lambda(X)>0,
\label{eq:strictpositive}
\end{equation}
where \(>0\) denotes positive definiteness. This property is also called positivity improving in the operator-scaling literature \cite{GeorgiouPavon2015,Filippov2021Sinkhorn}. Quantum analogues of classical Sinkhorn scaling relate strictly positive maps to doubly stochastic or unital representatives \cite{Sinkhorn1964,Gurvits2004,GeorgiouPavon2015,Cariello2019}. In particular, Theorem~5 of Georgiou and Pavon gives the doubly stochastic scaling for positivity-improving Kraus maps \cite{GeorgiouPavon2015}. Filippov summarizes the form needed below, including the inverse filtering relation, and traces it to the operator-scaling results of Gurvits and Georgiou--Pavon \cite{Gurvits2004,GeorgiouPavon2015,Filippov2021Sinkhorn}.

\begin{lemma}[Quantum Sinkhorn scaling]
\label{lem:sinkhorn}
Let \(\Lambda\) be a strictly positive CPTP map on a finite-dimensional matrix algebra. Then there exist positive-definite invertible operators \(A\) and \(B\) such that
\begin{equation}
    \Upsilon
    =\Phi_A\circ\Lambda\circ\Phi_B
\label{eq:sinkhorn}
\end{equation}
is trace preserving and unital. Equivalently,
\begin{equation}
    \Lambda
    =\Phi_{A^{-1}}\circ\Upsilon\circ\Phi_{B^{-1}}.
\label{eq:sinkhorninverse}
\end{equation}
\end{lemma}

For qubit channels, unitality permits a canonical unitary--Pauli--unitary form \cite{RuskaiSzarekWerner2002,ChoiLi2023}. This additional structure controls the channel transpose and is the only genuinely qubit-specific ingredient in the proof.

\subsection{Channel transposition and the factorization criterion}

Fix a qubit basis. If a CP map has Kraus representation
\begin{equation}
    \Lambda(X)=\sum_i K_i X K_i^\dagger,
\end{equation}
its channel transpose is the CP map
\begin{equation}
    \Lambda^T(X)=\sum_i K_i^T X \overline{K_i}.
\label{eq:channeltranspose}
\end{equation}
This definition is independent of the chosen Kraus representation, although it depends on the underlying basis \cite{MasajadaFellousStreltsov2026}. The basic maximally-entangled-state identity gives
\begin{equation}
J(\Lambda_2\circ\Lambda_1)
=
(\Lambda_1^T\otimes\Lambda_2)
\left(\lvert\Phi^+\rangle\!\langle\Phi^+\rvert\right).
\label{eq:choitransfer}
\end{equation}

The following criterion is the homogeneous conic form of the transpose-factorization result developed in Ref.~\cite{MasajadaFellousStreltsov2026}. Its short proof is included to make the logical dependence explicit.

\begin{lemma}[Transpose-factorization criterion]
\label{lem:factorcriterion}
Let \(\Lambda_1\) and \(\Lambda_2\) be qubit channels. Suppose that, in some basis,
\begin{equation}
    \Lambda_1^T
    =\mathcal{F}\circ\Lambda_1\circ\mathcal{E},
\label{eq:transposefactor}
\end{equation}
where \(\mathcal{F}\) is positive and \(\mathcal{E}\) is completely positive. If
\begin{equation}
    \Lambda_1\otimes\Lambda_2\in\EA,
\end{equation}
then
\begin{equation}
    \Lambda_2\circ\Lambda_1\in\EB.
\end{equation}
\end{lemma}

\begin{proof}
Set \(\phi^+=\lvert\Phi^+\rangle\!\langle\Phi^+\rvert\). Complete positivity of \(\mathcal{E}\) implies
\[
(\mathcal{E}\otimes\id)(\phi^+)\ge0.
\]
If \(\Lambda_1\otimes\Lambda_2\) is EA, then
\[
Z=(\Lambda_1\otimes\Lambda_2)
   \big[(\mathcal{E}\otimes\id)(\phi^+)\big]
\in\Sep.
\]
Lemma~\ref{lem:sepstability} and positivity of \(\mathcal{F}\) give
\[
(\mathcal{F}\otimes\id)(Z)\in\Sep.
\]
Using Eqs.~\eqref{eq:choitransfer} and \eqref{eq:transposefactor},
\[
(\mathcal{F}\otimes\id)(Z)
=
J(\Lambda_2\circ\Lambda_1).
\]
Thus the Choi state of \(\Lambda_2\circ\Lambda_1\) is separable, so the composition is EB \cite{HorodeckiShorRuskai2003}.
\end{proof}

\section{Main result}

The first step is to show that quantum Sinkhorn scaling supplies the transpose factorization required by Lemma~\ref{lem:factorcriterion} throughout the strictly positive interior of the qubit-channel set.

\begin{proposition}[Transpose factorization for strictly positive qubit channels]
\label{prop:strictfactor}
Let \(\Lambda\) be a strictly positive qubit channel. Then there exist invertible completely positive maps \(\mathcal{F}\) and \(\mathcal{E}\), whose inverses are also completely positive, such that
\begin{equation}
    \Lambda^T
    =\mathcal{F}\circ\Lambda\circ\mathcal{E}.
\label{eq:strictfactor}
\end{equation}
Thus \(\Lambda\) and \(\Lambda^T\) are equivalent under invertible local CP filtering. In particular, the weaker positive/CP factorization criterion of Lemma~\ref{lem:factorcriterion} holds without a restriction on Kraus rank.
\end{proposition}

\begin{proof}
By Lemma~\ref{lem:sinkhorn}, there exist positive-definite invertible \(A\) and \(B\) such that
\begin{equation}
    \Upsilon=\Phi_A\circ\Lambda\circ\Phi_B
\label{eq:upsilon}
\end{equation}
is a unital qubit channel. A unital qubit channel can be written as
\begin{equation}
    \Upsilon=\mathcal{U}\circ\mathcal{P}\circ\mathcal{V},
\label{eq:unitalnormal}
\end{equation}
where \(\mathcal{U}\) and \(\mathcal{V}\) are unitary channels and \(\mathcal{P}\) is a Pauli channel \cite{RuskaiSzarekWerner2002,ChoiLi2023,MasajadaFellousStreltsov2026}. Channel transposition reverses composition order. Moreover, \(\mathcal{P}^T=\mathcal{P}\), since transposition changes the sign of the \(\sigma_y\) Kraus operator but not the induced channel. Hence
\begin{equation}
    \Upsilon^T
    =\mathcal{V}^T\circ\mathcal{P}\circ\mathcal{U}^T.
\end{equation}
Define the unitary channels
\begin{equation}
    \mathcal{L}
    =\mathcal{V}^T\circ\mathcal{U}^{-1},
    \qquad
    \mathcal{R}
    =\mathcal{V}^{-1}\circ\mathcal{U}^T.
\label{eq:LR}
\end{equation}
Then
\begin{equation}
    \Upsilon^T
    =\mathcal{L}\circ\Upsilon\circ\mathcal{R}.
\label{eq:upsilontranspose}
\end{equation}

Taking the channel transpose of Eq.~\eqref{eq:sinkhorninverse} gives
\begin{equation}
\Lambda^T
=
\Phi_{\overline{B}^{-1}}
\circ\Upsilon^T
\circ\Phi_{\overline{A}^{-1}},
\label{eq:transposedsinkhorn}
\end{equation}
Since \(A\) and \(B\) are Hermitian positive definite, \(A^T=\overline{A}\) and \(B^T=\overline{B}\); hence \((A^{-1})^T=\overline{A}^{-1}\) and \((B^{-1})^T=\overline{B}^{-1}\). The matrices \(\overline{A}\) and \(\overline{B}\) are themselves positive definite and invertible, so the conjugated filters that appear below remain invertible CP maps with CP inverses. Inserting Eqs.~\eqref{eq:upsilontranspose} and \eqref{eq:upsilon} yields
\begin{align}
\Lambda^T
&=
\big(
\Phi_{\overline{B}^{-1}}
\circ\mathcal{L}\circ\Phi_A
\big)
\circ\Lambda
\circ
\big(
\Phi_B\circ\mathcal{R}\circ\Phi_{\overline{A}^{-1}}
\big).
\label{eq:factorconstructed}
\end{align}
Therefore Eq.~\eqref{eq:strictfactor} holds with
\begin{align}
\mathcal{F}
&=
\Phi_{\overline{B}^{-1}}
\circ\mathcal{L}\circ\Phi_A,
\\
\mathcal{E}
&=
\Phi_B\circ\mathcal{R}\circ\Phi_{\overline{A}^{-1}}.
\end{align}
Both maps are CP, so \(\mathcal{F}\) is in particular positive. Each factor in \(\mathcal{F}\) and \(\mathcal{E}\) is either a unitary channel or a filter \(\Phi_X\) with \(X\) invertible. Consequently, \(\mathcal{F}\) and \(\mathcal{E}\) are invertible and their inverses are CP. Equation~\eqref{eq:strictfactor} therefore gives a full equivalence between \(\Lambda\) and \(\Lambda^T\) under invertible CP pre- and postfilters.
\end{proof}

This proposition makes direct contact with the low-rank construction in Ref.~\cite{MasajadaFellousStreltsov2026}: that work obtained a factorization of \(\Lambda^T\) through \(\Lambda\) almost everywhere for Kraus rank at most three. Equation~\eqref{eq:factorconstructed} supplies such a factorization for every strictly positive qubit channel, including full-rank channels.

\begin{corollary}[Strictly positive first channel]
\label{cor:strict}
Let \(\Lambda_1\) and \(\Lambda_2\) be qubit channels, with \(\Lambda_1\) strictly positive. If
\begin{equation}
    \Lambda_1\otimes\Lambda_2\in\EA,
\end{equation}
then
\begin{equation}
    \Lambda_2\circ\Lambda_1\in\EB.
\end{equation}
\end{corollary}

\begin{proof}
Proposition~\ref{prop:strictfactor} provides the hypothesis of Lemma~\ref{lem:factorcriterion}.
\end{proof}

\begin{remark}[Alternative route through the filtered unital criterion]
For strictly positive $\Lambda_1$, Corollary~1 also follows from the unital-channel result of Ref.~\cite{MasajadaFellousStreltsov2026}, which allows a completely positive intermediate filter. Write $\Lambda_1 = \Phi_{A^{-1}} \circ \Upsilon \circ \Phi_{B^{-1}}$ as in Eq.~\eqref{eq:sinkhorninverse}. Moving the prefilter to the reference system gives
\begin{equation}
  J(\Lambda_2 \circ \Lambda_1) = (\Phi_{\overline{B}^{-1}} \otimes \mathrm{id}\bigr) \left[ J\!\left(\Lambda_2 \circ \Phi_{A^{-1}} \circ \Upsilon\right) \right] ,
\end{equation}
and since $\Phi_{\overline{B}^{-1}}$ is invertible with a completely positive inverse, one operator is separable exactly when the other is. If $\Lambda_2 \circ \Lambda_1 \notin \mathrm{EB}$, the filtered unital criterion with filter $\Phi_{A^{-1}}$ therefore gives $\Upsilon \otimes \Lambda_2 \notin \mathrm{EA}$. Because $\Upsilon \otimes \Lambda_2 = (\Phi_A \otimes \mathrm{id}) \circ (\Lambda_1 \otimes \Lambda_2) \circ (\Phi_B \otimes \mathrm{id})$, Lemma~1 shows that $\Lambda_1 \otimes \Lambda_2 \in \mathrm{EA}$ would force $\Upsilon \otimes \Lambda_2 \in \mathrm{EA}$, so $\Lambda_1 \otimes \Lambda_2 \notin \mathrm{EA}$. The transpose-factorization route of Proposition~1 is retained because it yields the stronger statement that $\Lambda$ and $\Lambda^T$ are equivalent under invertible completely positive filters, which is the structure used in Sec.~VI to formulate the higher-dimensional question.
\end{remark}

\subsection{Extension to arbitrary qubit channels}

Let \(\Dcal\) be the completely depolarizing qubit channel,
\begin{equation}
    \Dcal(X)=\frac{\operatorname{Tr}(X)}{2}I.
\label{eq:depolarizing}
\end{equation}
For \(0<\varepsilon<1\), define
\begin{equation}
    \Lambda_{1,\varepsilon}
    =(1-\varepsilon)\Lambda_1+\varepsilon\Dcal.
\label{eq:regularized}
\end{equation}
For every nonzero \(X\ge0\),
\begin{align}
\Lambda_{1,\varepsilon}(X)
&=(1-\varepsilon)\Lambda_1(X)
+\varepsilon\frac{\operatorname{Tr}(X)}{2}I
>0,
\end{align}
so \(\Lambda_{1,\varepsilon}\) is strictly positive.

\begin{theorem}[Parallel EA implies sequential EB]
\label{thm:main}
Let \(\Lambda_1\) and \(\Lambda_2\) be arbitrary qubit channels. If
\begin{equation}
    \Lambda_1\otimes\Lambda_2\in\EA,
\label{eq:mainpremise}
\end{equation}
then
\begin{equation}
    \Lambda_2\circ\Lambda_1\in\EB.
\label{eq:mainconclusion}
\end{equation}
\end{theorem}

\begin{proof}
Assume Eq.~\eqref{eq:mainpremise}. For any bipartite positive operator \(X_{AB}\),
\begin{equation}
(\Dcal\otimes\Lambda_2)(X_{AB})
=\frac{I_A}{2}\otimes
\Lambda_2\!\left(\operatorname{Tr}_A X_{AB}\right),
\label{eq:Dproduct}
\end{equation}
which is a positive product operator. Hence
\(\Dcal\otimes\Lambda_2\in\EA\).

Using Eq.~\eqref{eq:regularized},
\begin{align}
\Lambda_{1,\varepsilon}\otimes\Lambda_2
&=(1-\varepsilon)(\Lambda_1\otimes\Lambda_2)
+\varepsilon(\Dcal\otimes\Lambda_2).
\end{align}
The cone of separable positive operators is convex, and therefore the set of EA maps is convex. Since both terms on the right-hand side are EA,
\begin{equation}
\Lambda_{1,\varepsilon}\otimes\Lambda_2\in\EA
\quad\text{for every }\varepsilon\in(0,1).
\end{equation}
The channel \(\Lambda_{1,\varepsilon}\) is strictly positive, so Corollary~\ref{cor:strict} yields
\begin{equation}
\Lambda_2\circ\Lambda_{1,\varepsilon}\in\EB
\quad\text{for every }\varepsilon>0.
\label{eq:epsilonEB}
\end{equation}

As \(\varepsilon\to0\),
\begin{equation}
\Lambda_2\circ\Lambda_{1,\varepsilon}
\longrightarrow
\Lambda_2\circ\Lambda_1
\end{equation}
in any norm on the finite-dimensional space of superoperators. The Choi map is continuous, and a channel is EB if and only if its Choi state is separable \cite{HorodeckiShorRuskai2003}. The finite-dimensional set of separable density operators is closed. Hence the set of EB channels is closed, and the limit of Eq.~\eqref{eq:epsilonEB} gives Eq.~\eqref{eq:mainconclusion}.
\end{proof}

Taking the contrapositive gives the form directly relevant to source placement.

\begin{corollary}[Qualitative midpoint feasibility]
\label{cor:conjecture}
For arbitrary qubit channels \(\Lambda_1\) and \(\Lambda_2\),
\begin{equation}
\boxed{
\Lambda_2\circ\Lambda_1\notin\EB
\;\Longrightarrow\;
\Lambda_1\otimes\Lambda_2\notin\EA
}.
\label{eq:boxed}
\end{equation}
\end{corollary}

\begin{proof}
This is the contrapositive of Theorem~\ref{thm:main}.
\end{proof}

\begin{corollary}[Both endpoint orientations]
\label{cor:bothorders}
If \(\Lambda_1\otimes\Lambda_2\in\EA\), then both sequential compositions are EB:
\begin{equation}
\Lambda_2\circ\Lambda_1\in\EB,
\qquad
\Lambda_1\circ\Lambda_2\in\EB.
\end{equation}
Equivalently, if either endpoint placement preserves entanglement for some input state, then midpoint placement preserves entanglement for some input state.
\end{corollary}

\begin{proof}
The first statement is Theorem~\ref{thm:main}. If \(\Lambda_1\otimes\Lambda_2\) is EA, then \(\Lambda_2\otimes\Lambda_1\) is also EA because the two maps are related by unitary swaps of the input and output subsystems, which preserve separability. Applying Theorem~\ref{thm:main} to the ordered pair \((\Lambda_2,\Lambda_1)\) gives \(\Lambda_1\circ\Lambda_2\in\EB\).
\end{proof}

\begin{remark}[No factorwise entanglement-breaking conclusion]
\label{rem:no-factorwise}
The premise \(\Lambda_1\otimes\Lambda_2\in\EA\) does not imply that either local factor is EB. For qubit depolarizing channels
\begin{equation}
\mathcal{E}_{q}(X)=qX+(1-q)\operatorname{Tr}(X)\frac{I}{2},
\end{equation}
Section~III of Ref.~\cite{FilippovRybarZiman2012}, after the positivity analysis culminating in its Eq.~(3), derives for generally distinct depolarizing parameters that \(\mathcal{E}_{q_1}\otimes\mathcal{E}_{q_2}\) is EA if and only if \(q_1q_2\le 1/3\), whereas a single \(\mathcal{E}_{q}\) is EB if and only if \(q\le 1/3\). Thus \(q_1=q_2=1/2\) gives an EA parallel channel although neither factor is EB. The composition has depolarizing parameter \(q_1q_2=1/4\) and is EB, exactly as Theorem~\ref{thm:main} requires.
\end{remark}

\section{Consequences for entanglement-source placement}

Corollary~2 proves the midpoint-optimality conjecture of Ref.~\cite{MasajadaFellousStreltsov2026} in the feasibility sense defined there. If the source is placed at an endpoint and the transmitted half passes through \(\Lambda_1\) followed by \(\Lambda_2\), then entanglement can survive for some input exactly when \(\Lambda_2\circ\Lambda_1\) is not EB. Corollary~\ref{cor:conjecture} shows that in that case \(\Lambda_1\otimes\Lambda_2\) cannot be EA, so there exists at least one bipartite input whose entanglement survives midpoint transmission.

This result should not be conflated with the stronger quantitative numerical conjecture in Ref.~\cite{MasajadaFellousStreltsov2026}. That conjecture compares a partially-transposed Choi eigenvalue for the sequential channel with an SDP quantity associated with the midpoint channel. Corollary~\ref{cor:conjecture} proves only the binary implication ``endpoint feasible \(\Rightarrow\) midpoint feasible''; it does not establish a pointwise ordering of negativity, concurrence, fidelity, or any other entanglement measure. Optimal inputs for noisy qubit channels can in fact be nonmaximally entangled, particularly for nonunital noise \cite{Pal2014Singlet,FilippovFrizenKolobova2018,SiddhuSmolin2023Optimal}.

The proof also clarifies the role of the low-Kraus-rank analysis in Ref.~\cite{MasajadaFellousStreltsov2026}. Its transpose-factorization criterion is the direct precursor of Lemma~\ref{lem:factorcriterion}, while the rank-\(\le3\) construction is replaced here by the Sinkhorn factorization in Proposition~\ref{prop:strictfactor} and the boundary regularization in Theorem~\ref{thm:main}.

\section{Relation to previous EA, EB, and composition results}

EA and EB have developed as distinct but related channel classes \cite{MoravcikovaZiman2010,FilippovRybarZiman2012,FilippovZiman2013,HorodeckiShorRuskai2003,Ruskai2003QubitEB,LamiHuber2016}. Recent work extends these notions in different directions. Mallick, Ganguly, and Majumdar introduce partially entanglement-breaking and partially entanglement-annihilating channels to track reductions of Schmidt number rather than only the binary separable/entangled distinction \cite{MallickGangulyMajumdar2026}. La Piana and M\"uller-Hermes study Lorentz-entanglement-breaking and Lorentz-entanglement-annihilating maps through factorization and operator-ideal norms \cite{LaPianaMullerHermes2026}. Neither framework directly imposes the heterogeneous parallel-versus-sequential implication of Theorem~\ref{thm:main}.

Aubrun and M\"uller-Hermes place entanglement annihilation in a more general ordered-vector-space setting based on maximal and minimal tensor products of cones \cite{AubrunMullerHermes2023}. Their max-entanglement-annihilation condition is stronger than the standard physical EA condition used here: it is an all-tensor-powers requirement involving maximal cone tensor products. Because \(\mathrm{PSD}(\mathbb{C}^2)\) is order-isomorphic to a Lorentz cone, their resilience theorem implies that a qubit map satisfying that stronger max-EA property is itself EB. This does not subsume Theorem~\ref{thm:main}. Under ordinary physical EA, even \(\Lambda\otimes\Lambda\in\EA\) need not imply \(\Lambda\in\EB\), as the depolarizing example in Remark~\ref{rem:no-factorwise} shows; the conclusion available here is instead \(\Lambda^2\in\EB\). For heterogeneous pairs, the theorem analogously constrains the ordered composition without forcing either factor to be EB.

A neighboring literature studies EB behavior under repeated or structured composition. EB indices and entanglement-saving channels quantify persistence under powers \cite{LamiGiovannetti2015Indices,LamiGiovannetti2016Saving}, eventual EB behavior is tied to primitivity and long-time dynamics \cite{RahamanJaquesPaulsen2018,HansonRouzeFranca2020,AhiableEtAl2021}, and the PPT\(^2\) program asks when compositions of structured positive maps become EB \cite{KennedyManorPaulsen2018,ChristandlMullerHermesWolf2019,ChenYangTang2019,SinghNechita2022}. Theorem~\ref{thm:main} is of a different type: a parallel EA premise for two possibly different local channels forces a single ordered composition to be EB. The local-filter step is closely related to the use of invertible normal forms in Ref.~\cite{FilippovFrizenKolobova2018}, while the present transpose factorization combines that viewpoint with the source-placement criterion of Ref.~\cite{MasajadaFellousStreltsov2026}.

\section{Dimensional scope and outlook}

The depolarizing regularization and closedness arguments are not intrinsically restricted to qubits, and quantum Sinkhorn scaling has finite-dimensional formulations beyond \(d=2\) \cite{Gurvits2004,GeorgiouPavon2015,Cariello2019}. The qubit restriction enters in Proposition~\ref{prop:strictfactor}. For qubits, the conclusion is stronger than the asymmetric factorization required by Lemma~\ref{lem:factorcriterion}: every strictly positive qubit channel is equivalent to its channel transpose by invertible CP pre- and postfilters, and the same strong equivalence already holds for the unital representative through Eq.~\eqref{eq:upsilontranspose}.

This separates the higher-dimensional problem into two levels. The strong question is whether every unital qudit channel \(\Upsilon\) is CP-filter equivalent to its transpose, i.e., whether there exist invertible CP maps \(\mathcal{L}\) and \(\mathcal{R}\), with CP inverses, such that
\begin{equation}
    \Upsilon^T
    =\mathcal{L}\circ\Upsilon\circ\mathcal{R}.
\label{eq:quditstrong}
\end{equation}
The weaker question, which is sufficient for the source-placement theorem, asks only whether there exist a positive map \(\mathcal{F}\) and a completely positive map \(\mathcal{E}\) such that
\begin{equation}
    \Upsilon^T
    =\mathcal{F}\circ\Upsilon\circ\mathcal{E}.
\label{eq:quditquestion}
\end{equation}
Qubit unital channels satisfy the strong condition. A positive answer to the weak condition for any class of unital qudit channels would, after the same Sinkhorn and regularization steps, extend the midpoint-feasibility theorem to the corresponding class of arbitrary qudit channels. Failure of the strong condition would not by itself preclude such an extension; failure of the weak condition would identify the precise obstruction to this proof strategy.

The factorization viewpoint also suggests a broader strategy for channel-comparison problems: pass to a normal form by invertible filters, establish the desired relation for the normal form, and recover boundary channels through a convex regularization. Related filter-normal-form methods have proved useful in entanglement-robustness problems \cite{VerstraeteDehaeneDeMoor2003,FilippovFrizenKolobova2018}.

\section{Conclusion}

For arbitrary qubit channels \(\Lambda_1\) and \(\Lambda_2\), the main result establishes
\begin{equation}
\Lambda_1\otimes\Lambda_2\in\EA
\quad\Longrightarrow\quad
\Lambda_2\circ\Lambda_1\in\EB.
\end{equation}
The proof follows the transpose-factorization route introduced in Ref.~\cite{MasajadaFellousStreltsov2026}. Quantum Sinkhorn scaling first converts a strictly positive channel to a unital qubit representative; the unital normal form then shows that the transposed representative is obtained by unitary pre- and postprocessing, producing the required factorization through the original channel. Depolarizing regularization extends the result from the strictly positive interior to the full qubit-channel set.

By contraposition, any entanglement that can survive an endpoint configuration can also survive the midpoint configuration for a suitable input state, and exchanging the two channels gives the same conclusion for the opposite endpoint. This proves the midpoint-optimality conjecture of Ref.~\cite{MasajadaFellousStreltsov2026} in its original feasibility sense. The stronger quantitative SDP inequality conjectured in that work remains open. In higher dimensions, the strong CP-filter-equivalence question in Eq.~\eqref{eq:quditstrong} and the weaker sufficient factorization in Eq.~\eqref{eq:quditquestion} isolate two concrete levels at which the qubit argument may or may not extend.

\bibliographystyle{apsrev4-2}
\bibliography{referencias}

@article{MasajadaFellousStreltsov2026,
  author        = {Piotr Masajada and Marco Fellous-Asiani and Alexander Streltsov},
  title         = {Optimizing entanglement distribution via noisy quantum channels},
  journal       = {Phys. Rev. A},
  volume        = {113},
  pages         = {052414},
  year          = {2026},
  doi           = {10.1103/hdzn-fwpj},
  eprint        = {2506.06089},
  archivePrefix = {arXiv},
  primaryClass  = {quant-ph}
}

@article{Pal2014Singlet,
  author  = {Rajarshi Pal and Somshubhro Bandyopadhyay and Sibasish Ghosh},
  title   = {Entanglement sharing through noisy qubit channels: One-shot optimal singlet fraction},
  journal = {Phys. Rev. A},
  volume  = {90},
  pages   = {052304},
  year    = {2014},
  doi     = {10.1103/PhysRevA.90.052304}
}

@article{Streltsov2015Unified,
  author  = {Alexander Streltsov and Remigiusz Augusiak and Maciej Demianowicz and Maciej Lewenstein},
  title   = {Progress towards a unified approach to entanglement distribution},
  journal = {Phys. Rev. A},
  volume  = {92},
  pages   = {012335},
  year    = {2015},
  doi     = {10.1103/PhysRevA.92.012335}
}

@article{Zuppardo2016Excessive,
  author  = {Margherita Zuppardo and Tanjung Krisnanda and Tomasz Paterek and Somshubhro Bandyopadhyay and Anindita Banerjee and Prasenjit Deb and Saronath Halder and Kavan Modi and Mauro Paternostro},
  title   = {Excessive distribution of quantum entanglement},
  journal = {Phys. Rev. A},
  volume  = {93},
  pages   = {012305},
  year    = {2016},
  doi     = {10.1103/PhysRevA.93.012305}
}

@article{Streltsov2012QuantumCost,
  author  = {Alexander Streltsov and Hermann Kampermann and Dagmar Bru{\ss}},
  title   = {Quantum Cost for Sending Entanglement},
  journal = {Phys. Rev. Lett.},
  volume  = {108},
  pages   = {250501},
  year    = {2012},
  doi     = {10.1103/PhysRevLett.108.250501}
}

@article{Chuan2012Discord,
  author  = {T. K. Chuan and J. Maillard and K. Modi and T. Paterek and M. Paternostro and M. Piani},
  title   = {Quantum Discord Bounds the Amount of Distributed Entanglement},
  journal = {Phys. Rev. Lett.},
  volume  = {109},
  pages   = {070501},
  year    = {2012},
  doi     = {10.1103/PhysRevLett.109.070501}
}

@article{Cubitt2003Separable,
  author  = {T. S. Cubitt and F. Verstraete and W. D{\"u}r and J. I. Cirac},
  title   = {Separable States Can Be Used To Distribute Entanglement},
  journal = {Phys. Rev. Lett.},
  volume  = {91},
  pages   = {037902},
  year    = {2003},
  doi     = {10.1103/PhysRevLett.91.037902}
}

@article{SiddhuSmolin2023Optimal,
  author  = {Vikesh Siddhu and John Smolin},
  title   = {Optimal one-shot entanglement sharing},
  journal = {Phys. Rev. A},
  volume  = {108},
  pages   = {032617},
  year    = {2023},
  doi     = {10.1103/PhysRevA.108.032617}
}

@article{MoravcikovaZiman2010,
  author  = {Lenka Morav{\v{c}}{\'i}kov{\'a} and M{\'a}rio Ziman},
  title   = {Entanglement-annihilating and entanglement-breaking channels},
  journal = {J. Phys. A: Math. Theor.},
  volume  = {43},
  pages   = {275306},
  year    = {2010},
  doi     = {10.1088/1751-8113/43/27/275306}
}

@article{FilippovRybarZiman2012,
  author  = {Sergey N. Filippov and Tom{\'a}{\v{s}} Ryb{\'a}r and M{\'a}rio Ziman},
  title   = {Local two-qubit entanglement-annihilating channels},
  journal = {Phys. Rev. A},
  volume  = {85},
  pages   = {012303},
  year    = {2012},
  doi     = {10.1103/PhysRevA.85.012303}
}

@article{FilippovZiman2013,
  author  = {Sergey N. Filippov and M{\'a}rio Ziman},
  title   = {Bipartite entanglement-annihilating maps: Necessary and sufficient conditions},
  journal = {Phys. Rev. A},
  volume  = {88},
  pages   = {032316},
  year    = {2013},
  doi     = {10.1103/PhysRevA.88.032316}
}

@article{HorodeckiShorRuskai2003,
  author  = {Michael Horodecki and Peter W. Shor and Mary Beth Ruskai},
  title   = {Entanglement Breaking Channels},
  journal = {Rev. Math. Phys.},
  volume  = {15},
  pages   = {629--641},
  year    = {2003},
  doi     = {10.1142/S0129055X03001709}
}

@article{Ruskai2003QubitEB,
  author  = {Mary Beth Ruskai},
  title   = {Qubit Entanglement Breaking Channels},
  journal = {Rev. Math. Phys.},
  volume  = {15},
  pages   = {643--662},
  year    = {2003},
  doi     = {10.1142/S0129055X03001710}
}

@article{LamiGiovannetti2015Indices,
  author  = {Ludovico Lami and Vittorio Giovannetti},
  title   = {Entanglement-Breaking Indices},
  journal = {J. Math. Phys.},
  volume  = {56},
  pages   = {092201},
  year    = {2015},
  doi     = {10.1063/1.4931482}
}

@article{LamiGiovannetti2016Saving,
  author  = {Ludovico Lami and Vittorio Giovannetti},
  title   = {Entanglement-Saving Channels},
  journal = {J. Math. Phys.},
  volume  = {57},
  pages   = {032201},
  year    = {2016},
  doi     = {10.1063/1.4942495}
}

@article{LamiHuber2016,
  author  = {Ludovico Lami and Marcus Huber},
  title   = {Bipartite depolarizing maps},
  journal = {J. Math. Phys.},
  volume  = {57},
  pages   = {092201},
  year    = {2016},
  doi     = {10.1063/1.4962339}
}

@article{FilippovFrizenKolobova2018,
  author  = {Sergey N. Filippov and Vladimir V. Frizen and Daria V. Kolobova},
  title   = {Ultimate entanglement robustness of two-qubit states against general local noises},
  journal = {Phys. Rev. A},
  volume  = {97},
  pages   = {012322},
  year    = {2018},
  doi     = {10.1103/PhysRevA.97.012322}
}

@article{AubrunMullerHermes2023,
  author  = {Guillaume Aubrun and Alexander M{\"u}ller-Hermes},
  title   = {Annihilating Entanglement Between Cones},
  journal = {Commun. Math. Phys.},
  volume  = {400},
  pages   = {931--976},
  year    = {2023},
  doi     = {10.1007/s00220-022-04621-5}
}

@article{MallickGangulyMajumdar2026,
  author  = {Bivas Mallick and Nirman Ganguly and Archan S. Majumdar},
  title   = {On the characterization of partially entanglement breaking and annihilating channels},
  journal = {J. Math. Phys.},
  volume  = {67},
  pages   = {082201},
  year    = {2026},
  doi     = {10.1063/5.0274921}
}

@article{LaPianaMullerHermes2026,
  author  = {Francesca La Piana and Alexander M{\"u}ller-Hermes},
  title   = {Annihilating and breaking Lorentz cone entanglement},
  journal = {Linear Algebra Appl.},
  volume  = {739},
  pages   = {68--103},
  year    = {2026},
  doi     = {10.1016/j.laa.2026.03.012}
}

@article{Sinkhorn1964,
  author  = {Richard Sinkhorn},
  title   = {A Relationship Between Arbitrary Positive Matrices and Doubly Stochastic Matrices},
  journal = {Ann. Math. Stat.},
  volume  = {35},
  pages   = {876--879},
  year    = {1964},
  doi     = {10.1214/aoms/1177703591}
}

@article{Gurvits2004,
  author  = {Leonid Gurvits},
  title   = {Classical complexity and quantum entanglement},
  journal = {J. Comput. Syst. Sci.},
  volume  = {69},
  pages   = {448--484},
  year    = {2004},
  doi     = {10.1016/j.jcss.2004.06.003}
}

@article{GeorgiouPavon2015,
  author  = {Tryphon T. Georgiou and Michele Pavon},
  title   = {Positive contraction mappings for classical and quantum Schr{\"o}dinger systems},
  journal = {J. Math. Phys.},
  volume  = {56},
  pages   = {033301},
  year    = {2015},
  doi     = {10.1063/1.4915289}
}

@article{Filippov2021Sinkhorn,
  author  = {Sergey N. Filippov},
  title   = {Entanglement Robustness in Trace Decreasing Quantum Dynamics Caused by Depolarization and Polarization Dependent Losses},
  journal = {Quanta},
  volume  = {10},
  pages   = {15--21},
  year    = {2021},
  doi     = {10.12743/quanta.v10i1.163},
}

@article{VerstraeteDehaeneDeMoor2001,
  author  = {Frank Verstraete and Jeroen Dehaene and Bart De Moor},
  title   = {Local filtering operations on two qubits},
  journal = {Phys. Rev. A},
  volume  = {64},
  pages   = {010101},
  year    = {2001},
  doi     = {10.1103/PhysRevA.64.010101}
}

@article{VerstraeteDehaeneDeMoor2003,
  author  = {Frank Verstraete and Jeroen Dehaene and Bart De Moor},
  title   = {Normal forms and entanglement measures for multipartite quantum states},
  journal = {Phys. Rev. A},
  volume  = {68},
  pages   = {012103},
  year    = {2003},
  doi     = {10.1103/PhysRevA.68.012103}
}

@article{ChoiLi2023,
  author  = {Man-Duen Choi and Chi-Kwong Li},
  title   = {On unital qubit channels},
  journal = {Quantum Inf. Comput.},
  volume  = {23},
  number  = {7\&8},
  pages   = {562--576},
  year    = {2023},
  doi     = {10.26421/QIC23.7-8-2}
}

@article{RuskaiSzarekWerner2002,
  author  = {Mary Beth Ruskai and Stanislaw Szarek and Elisabeth Werner},
  title   = {An analysis of completely-positive trace-preserving maps on $2\times2$ matrices},
  journal = {Linear Algebra Appl.},
  volume  = {347},
  pages   = {159--187},
  year    = {2002},
  doi     = {10.1016/S0024-3795(01)00547-X}
}

@article{BraunEtAl2014,
  author  = {Daniel Braun and Olivier Giraud and Ion Nechita and Cl{\'e}ment Pellegrini and Marko {\v{Z}}nidari{\v{c}}},
  title   = {A universal set of qubit quantum channels},
  journal = {J. Phys. A: Math. Theor.},
  volume  = {47},
  pages   = {135302},
  year    = {2014},
  doi     = {10.1088/1751-8113/47/13/135302}
}

@article{Cariello2019,
  author  = {Daniel Cariello},
  title   = {Sinkhorn--Knopp theorem for rectangular positive maps},
  journal = {Linear Multilinear Algebra},
  volume  = {67},
  pages   = {2345--2365},
  year    = {2019},
  doi     = {10.1080/03081087.2018.1491524}
}

@article{FujiwaraAlgoet1999,
  author  = {Akio Fujiwara and Paul Algoet},
  title   = {One-to-one parametrization of quantum channels},
  journal = {Phys. Rev. A},
  volume  = {59},
  pages   = {3290--3294},
  year    = {1999},
  doi     = {10.1103/PhysRevA.59.3290}
}

@article{Jamiolkowski1972,
  author  = {Andrzej Jamio{\l}kowski},
  title   = {Linear transformations which preserve trace and positive semidefiniteness of operators},
  journal = {Rep. Math. Phys.},
  volume  = {3},
  pages   = {275--278},
  year    = {1972},
  doi     = {10.1016/0034-4877(72)90011-0}
}

@article{Choi1975,
  author  = {Man-Duen Choi},
  title   = {Completely positive linear maps on complex matrices},
  journal = {Linear Algebra Appl.},
  volume  = {10},
  pages   = {285--290},
  year    = {1975},
  doi     = {10.1016/0024-3795(75)90075-0}
}

@article{Peres1996,
  author  = {Asher Peres},
  title   = {Separability Criterion for Density Matrices},
  journal = {Phys. Rev. Lett.},
  volume  = {77},
  pages   = {1413--1415},
  year    = {1996},
  doi     = {10.1103/PhysRevLett.77.1413}
}

@article{Horodecki1996,
  author  = {Micha{\l} Horodecki and Pawe{\l} Horodecki and Ryszard Horodecki},
  title   = {Separability of mixed states: necessary and sufficient conditions},
  journal = {Phys. Lett. A},
  volume  = {223},
  pages   = {1--8},
  year    = {1996},
  doi     = {10.1016/S0375-9601(96)00706-2}
}

@article{VidalWerner2002,
  author  = {Guifr{\'e} Vidal and Reinhard F. Werner},
  title   = {Computable measure of entanglement},
  journal = {Phys. Rev. A},
  volume  = {65},
  pages   = {032314},
  year    = {2002},
  doi     = {10.1103/PhysRevA.65.032314}
}

@article{HorodeckiReview2009,
  author  = {Ryszard Horodecki and Pawe{\l} Horodecki and Micha{\l} Horodecki and Karol Horodecki},
  title   = {Quantum entanglement},
  journal = {Rev. Mod. Phys.},
  volume  = {81},
  pages   = {865--942},
  year    = {2009},
  doi     = {10.1103/RevModPhys.81.865}
}

@article{AhiableEtAl2021,
  author  = {Jennifer Ahiable and David W. Kribs and Jeremy Levick and Rajesh Pereira and Mizanur Rahaman},
  title   = {Entanglement breaking channels, stochastic matrices, and primitivity},
  journal = {Linear Algebra Appl.},
  volume  = {629},
  pages   = {219--231},
  year    = {2021},
  doi     = {10.1016/j.laa.2021.08.013}
}

@article{WolfCirac2008,
  author  = {Michael M. Wolf and J. Ignacio Cirac},
  title   = {Dividing Quantum Channels},
  journal = {Commun. Math. Phys.},
  volume  = {279},
  pages   = {147--168},
  year    = {2008},
  doi     = {10.1007/s00220-008-0411-y}
}

@article{KennedyManorPaulsen2018,
  author  = {Matthew Kennedy and Nicholas A. Manor and Vern I. Paulsen},
  title   = {Compositions of PPT Maps},
  journal = {Quantum Inf. Comput.},
  volume  = {18},
  number  = {5\&6},
  pages   = {472--480},
  year    = {2018},
}

@article{ChristandlMullerHermesWolf2019,
  author  = {Matthias Christandl and Alexander M{\"u}ller-Hermes and Michael M. Wolf},
  title   = {When Do Composed Maps Become Entanglement Breaking?},
  journal = {Ann. Henri Poincar{\'e}},
  volume  = {20},
  pages   = {2295--2322},
  year    = {2019},
  doi     = {10.1007/s00023-019-00774-7}
}

@article{ChenYangTang2019,
  author  = {Lin Chen and Yu Yang and Wai-Shing Tang},
  title   = {Positive-partial-transpose square conjecture for $n=3$},
  journal = {Phys. Rev. A},
  volume  = {99},
  pages   = {012337},
  year    = {2019},
  doi     = {10.1103/PhysRevA.99.012337}
}

@article{SinghNechita2022,
  author  = {Satvik Singh and Ion Nechita},
  title   = {The PPT$^2$ Conjecture Holds for All Choi-Type Maps},
  journal = {Ann. Henri Poincar{\'e}},
  volume  = {23},
  pages   = {3311--3329},
  year    = {2022},
  doi     = {10.1007/s00023-022-01166-0}
}

@article{RahamanJaquesPaulsen2018,
  author  = {Mizanur Rahaman and Samuel Jaques and Vern I. Paulsen},
  title   = {Eventually Entanglement Breaking Maps},
  journal = {J. Math. Phys.},
  volume  = {59},
  pages   = {062201},
  year    = {2018},
  doi     = {10.1063/1.5024385}
}

@article{HansonRouzeFranca2020,
  author  = {Eric P. Hanson and Cambyse Rouz{\'e} and Daniel Stilck Fran{\c{c}}a},
  title   = {Eventually Entanglement Breaking Markovian Dynamics: Structure and Characteristic Times},
  journal = {Ann. Henri Poincar{\'e}},
  volume  = {21},
  pages   = {1517--1571},
  year    = {2020},
  doi     = {10.1007/s00023-020-00906-4}
}

\end{document}